%% file: main.tex
\documentclass[conference,a4paper]{IEEEtran}
\IEEEoverridecommandlockouts
\usepackage[utf8]{inputenc}
\usepackage[T1]{fontenc}
\usepackage{cite}
\usepackage[cmex10]{amsmath}
\usepackage{amssymb,amsfonts,amsthm}
\usepackage{graphicx}
\usepackage{balance}
\usepackage{xcolor}
\usepackage{xspace}
\usepackage[capitalize]{cleveref}
\usepackage{algorithm}
\usepackage{enumitem}
\usepackage{pgfplots}
\pgfplotsset{compat=1.18}
\usetikzlibrary{shapes.geometric,shapes.symbols,arrows,positioning}

\newtheorem{theorem}{Theorem}

\newtheorem{lemma}{Lemma}

\newtheorem{definition}{Definition}

\newtheorem{remark}{Remark}

\input{commands}

\begin{document}

\title{Interactive Byzantine-Resilient Gradient Coding for General Data Assignments}

\author{
    \IEEEauthorblockN{Shreyas Jain, Luis Maßny, Christoph Hofmeister, Eitan Yaakobi, and Rawad Bitar}
    \thanks{SJ is with the Department of Mathematical Sciences at the Indian Institute of Science Education and Research, Mohali, India. Email: ms20098@iisermohali.ac.in}
    \thanks{CH, LM and RB are with the School of Computation, Information and Technology at the Technical University of Munich, Germany. Emails: \{christoph.hofmeister, luis.massny, rawad.bitar\}@tum.de}
    \thanks{EY is with the CS department of Technion---Israel Institute of Technology, Israel. Email: yaakobi@cs.technion.ac.il}
    \thanks{SJ was supported in part by a DAAD (German Academic Exchange Service)-WISE scholarship. This project is funded by the Bavarian Ministry of Economic Affairs, Regional Development and Energy within the scope of the 6G Future Lab Bavaria, and by DFG (German Research Foundation) projects under Grant Agreement No. WA 3907/7-1 and No. BI 2492/1-1.}
}

\maketitle

\begin{abstract}
We tackle the problem of Byzantine errors in distributed gradient descent within the Byzantine-resilient gradient coding framework. Our proposed solution can recover the exact full gradient in the presence of $\nmalicious$ malicious workers with a data replication factor of only $\nmalicious+1$. It generalizes previous solutions to any data assignment scheme that has a regular replication over all data samples. The scheme detects malicious workers through additional interactive communication and a small number of local computations at the main node, leveraging group-wise comparisons between workers with a provably optimal grouping strategy. The scheme requires at most $\nmalicious$ interactive rounds that incur a total communication cost logarithmic in the number of data samples.
\end{abstract}

\section{Introduction}
\label{sec:introduction}
\input{introduction}

\section{Problem Setting}
\label{sec:problem-setting}
\input{problem_setting}

\section{Scheme Construction for Arbitrary Allocations}
\label{sec:scheme}
\input{scheme_construction}

\section{Analysis and Discussion}
\label{sec:analysis}
\input{analysis}

\clearpage
\balance
\bibliographystyle{IEEEtran}

\clearpage
\appendix
\input{appendix}

\end{document}

%% file: commands.tex
\newcommand{\R}{\ensuremath{\mathbb{R}}}
\newcommand{\C}{\ensuremath{\mathbb{C}}}
\newcommand{\setstyle}[1]{\ensuremath{\mathcal{#1}}}

\newcommand{\nworker}{\ensuremath{n}\xspace}
\newcommand{\ngroup}{\ensuremath{m}\xspace}
\newcommand{\ngrad}{\ensuremath{p}\xspace}
\newcommand{\graddim}{\ensuremath{d}\xspace}
\newcommand{\tgrad}{\ensuremath{\mathbf{g}}\xspace}
\newcommand{\gestim}{\ensuremath{\mathbf{\widehat{g}}}\xspace}
\newcommand{\pgrad}[1][\gind]{\ensuremath{\mathbf{g}_{#1}}\xspace}

\newcommand{\gind}{\ensuremath{i}\xspace}
\newcommand{\wind}{\ensuremath{j}\xspace}
\newcommand{\groupind}{\ensuremath{k}\xspace}

\newcommand{\fieldsize}{\ensuremath{q}\xspace}
\newcommand{\galpha}{\ensuremath{\mathbb{F}}\xspace}

\newcommand{\nmalicious}{\ensuremath{s}\xspace}
\newcommand{\encfun}[1][]{\ensuremath{\operatorname{enc}_{#1}}\xspace}
\newcommand{\encind}{\ensuremath{e}\xspace}

\newcommand{\encfunset}{\ensuremath{\setstyle{E}}\xspace}
\newcommand{\decfun}[1][]{\ensuremath{\operatorname{dec}_{#1}}\xspace}
\newcommand{\nencfun}{\ensuremath{v}\xspace}
\newcommand{\replfact}[1][]{\ensuremath{\rho_{\textrm{#1}}}\xspace}
\newcommand{\commoh}[1][]{\ensuremath{\kappa_{\textrm{#1}}}\xspace}
\newcommand{\proto}{\ensuremath{\Pi}\xspace}
\newcommand{\nround}{\ensuremath{T}\xspace}
\newcommand{\localcomp}[1][]{\ensuremath{c_{#1}}\xspace}
\newcommand{\gcomp}[1][]{\ensuremath{\gamma}}

\newcommand{\params}{\ensuremath{\boldsymbol{\theta}}\xspace}

\newcommand{\loss}{\ensuremath{\ell}\xspace}
\newcommand{\gditer}{\ensuremath{\tau}\xspace}
\newcommand{\range}[1]{\ensuremath{[#1]}\xspace}
\newcommand{\roundind}{\ensuremath{t}\xspace}
\newcommand{\respdim}{\ensuremath{\zeta}\xspace}
\newcommand{\iresp}{\ensuremath{\mathbf{z}}\xspace}
\newcommand{\noisyiresp}{\ensuremath{\mathbf{\widetilde{z}}}\xspace}
\newcommand{\err}{\ensuremath{\mathbf{e}}\xspace}

\newcommand{\allocmat}[1][]{\ensuremath{\mathbf{A}_{#1}}}

\newcommand{\card}[1]{{\left\vert #1 \right\vert}}

\newcommand{\defeq}{\mathrel{\mathop:}=}
\newcommand{\BGC}{\text{$\nmalicious$-BGC}\xspace}
\newcommand{\N}{\mathbb{N}}
\newcommand{\nhonest}{\ensuremath{u}}

\newcommand{\tp}{{\mathsf{T}}}
\newcommand{\zero}{\ensuremath{\boldsymbol{0}}}
\newcommand{\one}{\ensuremath{\boldsymbol{1}}}
\newcommand{\G}{\ensuremath{\mathbf{G}}\xspace}
\newcommand{\W}{\ensuremath{\mathbf{W}}\xspace}
\newcommand{\Z}{\ensuremath{\mathbf{Z}}\xspace}
\newcommand{\noisyZ}{\ensuremath{\mathbf{\widetilde{Z}}}\xspace}
\newcommand{\B}{\ensuremath{\mathbf{B}}\xspace}
\newcommand{\E}{\ensuremath{\mathbf{E}}\xspace}
\newcommand{\MDS}{\ensuremath{\mathbf{F}}\xspace}
\newcommand{\rth}{\ensuremath{r}\xspace}
\newcommand{\groupset}[1][\groupind]{\ensuremath{\mathcal{G}_{#1}}}
\newcommand{\combvec}[1][\groupind]{\ensuremath{\mathbf{b}_{#1}}\xspace}
\newcommand{\assignmentset}{\ensuremath{{\mathcal{R}_{\replfact}}}}

%% file: introduction.tex
Distributed training can speed up machine learning procedures and enable the training for huge datasets that exceed the local memory~\cite{abadi2016tensorflow}.
In a distributed learning setting, a central server, referred to as the \emph{main node}, distributes the computation load to a set of client machines, referred to as the \emph{workers}.
A popular algorithm is (stochastic) gradient descent, which is an iterative procedure. In every iteration, each worker computes a gradient of a certain loss function over a local dataset and sends it to the main node. The naive aggregation of the gradients, e.g., by averaging, is vulnerable to corrupt gradients from malicious workers~\cite{damaskinos2019aggregathor,blanchardMachineLearningAdversaries}, which are often modeled as worst-case errors caused by a Byzantine adversary~\cite{lamportByzantineGeneralsProblem}. 

Several coding-theoretic methods have emerged to tackle the problem of Byzantine adversaries in distributed computing~\cite{hofmeisterSecurePrivateAdaptive2022,tangAdaptiveVerifiableCoded2022,dataDataEncodingByzantineResilient2021,solankiNonColludingAttacksIdentification2019,keshtkarjahromiSecureCodedCooperative2019,subramaniam2019,yuLagrangeCodedComputing2019} for linear and polynomial function computations. However, the loss functions of machine learning algorithms are often highly nonlinear. Hence, a large body of the literature instead focuses on either the design of aggregation algorithms that are resilient to corruption from a limited number of workers~\cite{yinByzantineRobust2018,chenDistributedStatisticalMachine2017,alistarhByzantineStochastic2018,sohn2020ElectionCoding}
or on outlier detection methods~\cite{blanchardMachineLearningAdversaries,guerraoui2018hidden,el2020fast,xieZenoDistributedStochastic,rajputDETOXRedundancybasedFramework,konstantinidisRobustDetection,cao2019DGD}. These methods are applicable for general non-linear loss functions. However, they can only recover an approximation of the correct computation result, which reduces the learning speed, see e.g.,~\cite{bitarStochasticGradientCoding2020} and references therein, and might perform poorly when the distribution of the training data is not identical among the workers~\cite{chenRevisitingDistributedSynchronous2017,tandon}.

To exactly recover the gradient while using general non-linear loss functions, a coding-theoretic method mitigating malicious workers is proposed in~\cite{draco} and is extended to communication-efficiency in~\cite{chen2021solon}.
This method is based on the gradient coding framework~\cite{tandon}. The main idea is to replicate local datasets among the workers and exploit the computational redundancy to allow for error correction by coding over the computation results of the workers. However, this method requires a high replication and thus increases the computational overhead significantly.

In~\cite{hofmeisterTradingCommunication2023}, the authors introduced a general Byzantine-resilient gradient coding framework that halves the redundancy. The key idea is to identify malicious workers 
(effectively transforming errors into erasures) by allowing local computations at the main node and introducing an interactive communication protocol between the workers and the main node. For a fractional repetition data assignment, i.e., the workers are divided into groups that have the same local data sets, \cite{hofmeisterTradingCommunication2023} proposes a scheme that performs pair-wise comparisons of the workers' computation results, and resolves conflicts by local computations at the main node to identify malicious workers.

In this work, we adopt the framework of~\cite{hofmeisterTradingCommunication2023} but tackle the problem of Byzantine-resilient gradient coding for general data assignments.
We give a construction for a Byzantine-resilient gradient coding scheme. While the scheme adopts the idea of comparing partial results between workers and resolving conflicts by interactive messaging, a generalized strategy is developed to compare the results. In contrast to fractional repetition data assignments, pair-wise comparisons are not possible for general data assignments. Therefore, groups of workers are compared instead. We design a method to group the workers and compare their computations. We show through a lower bound on the number of group comparisons that our grouping strategy is optimal.

%% file: problem_setting.tex
\paragraph*{Notation}
Scalars are denoted by
lower-case letters. Column vectors and matrices are denoted by bold lower-case and bold upper-case letters, respectively. Sets are denoted by calligraphic letters $\mathcal{A}$. $\mathbf{A}_{i,j}$ refers to the element in row $i$ and column $j$ of the matrix $\mathbf{A}$ . $\mathbf{A}_{\mathcal{R},\mathcal{C}}$ refers to the submatrix of $\mathbf{A}$ restricted to the rows indexed by $\mathcal{R} \subset \N$ and columns indexed by $\mathcal{C} \subset \N$. The matrices $\mathbf{A}_{\cdot,\mathcal{C}}$ and $\mathbf{A}_{\mathcal{R},\cdot}$ denote respectively the submatrix of $\mathbf{A}$ restricted to the columns in $\mathcal{C}$ and to the rows in $\mathcal{R}$. For an integer $a \geq 1$, let $\range{a} \defeq \left\{ 1,2,\dots,a \right\}$. Finally, $\one_{m \times n}$ and $\zero_{m \times n}$ denote the all-one and all-zero matrices of dimension $m \times n$. All proofs not presented in the main paper are deferred to the appendix.

We consider machine learning tasks that are formulated as an optimization problem of the form
$\min_{\params \in \R^d} \sum_{\gind=1}^{\ngrad} \loss(\params, x_\gind)$
over the dataset $\{x_1,\dots,x_\ngrad\}$ consisting of $\ngrad$ samples and a sample-wise loss function $\loss(\params, x_\gind)$. The parameters to be learned are arranged in a vector $\params \in \R^\graddim$. A gradient descent algorithm~\cite[Section 9.3]{boyd2004convex} is employed, such that the optimal solution $\params^\star$ is approached iteratively by computing for $\gditer>0$
\begin{equation*}
    \params_{\gditer+1} = \params_\gditer - \eta \sum_{\gind=1}^{\ngrad} \nabla\loss(\params_\gditer, x_\gind),
\end{equation*}
where the value $\sum_{\gind=1}^{\ngrad} \nabla\loss(\params_\gditer, x_\gind)$ is the gradient of the loss function evaluated at all the samples of the dataset. The first iterate $\params_0$ is chosen randomly.

Consider a distributed computing setting in which a number $\nworker$ of workers is hired to compute the partial gradients in a distributed manner. The main node coordinates the procedure. Each worker receives a subset of the samples and computes the respective partial gradients.
Widely used assignment schemes are fractional repetition and cyclic repetition~\cite{tandon}. Following~\cite{draco,chen2021solon,hofmeisterTradingCommunication2023}, this work studies the problem of reconstructing the gradient $\sum_{\gind=1}^{\ngrad} \nabla\loss(\params_\gditer, x_\gind)$ \emph{exactly} at the main node in every gradient descent iteration in the presence of $\nmalicious < \nworker$ malicious workers which are under the control of a Byzantine adversary. We focus on one iteration of the gradient descent algorithm; hence, we omit the iteration index $\gditer$.

In the sequel, we represent the gradients over a field $\mathbb{F}$ with a large (and potentially infinite) number of elements. More specifically, $\mathbb{F}$ can represent either the finite field $\mathbb{F}_\fieldsize$ with $\fieldsize > n$ and $\fieldsize \gg 1$ elements, the real numbers $\mathbb{R}$, or the complex numbers $\mathbb{C}$. 
Treating $\mathbb{F}$ abstractly avoids the problems of quantization, compression, and numerical stability, which, although relevant, are outside the scope of this current work.

We refer to the value $\pgrad^{(\gditer)} \in \galpha^\graddim$ as a \emph{partial gradient} at iteration $\gditer$, which is the quantized representation of $\nabla\loss(\params_\gditer, x_\gind)$, and we define $\tgrad^{(\gditer)} \defeq \sum_{\gind=1}^{\ngrad} \pgrad^{(\gditer)}$ as the \emph{full gradient}.
\begin{figure}
\centering
\resizebox{0.7\linewidth}{!}{\input{setting_example}}
\vspace{-0.2cm}
\caption{Gradient Coding: Each worker is assigned two partial gradients to compute and transmits a linear combination to the main node. Without malicious workers, the main node obtains the full gradient from the responses of any two workers. For the considered setting ($W_3$ is malicious), no existing gradient coding scheme can recover the exact gradient correctly. Our goal is to present a scheme that allows the main node to identify the malicious worker and reconstruct the full gradient correctly from the honest workers.}
\label{example}
\end{figure}
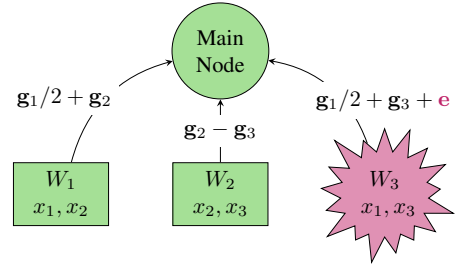
Our solution to the problem follows the intuition built from \cref{example}. We study our solution within the Byzantine-resilient gradient coding framework~\cite{hofmeisterTradingCommunication2023}, summarized below:
\begin{definition}[\BGC scheme~\cite{hofmeisterTradingCommunication2023}]
    \label{def:bgc}
    A Byzantine-resilient gradient coding scheme tolerating $\nmalicious$ malicious workers, referred to as \BGC scheme, is a tuple $\left( \allocmat, \encfunset, \decfun, \proto \right)$ consisting of
\begin{itemize}
    \item a \textbf{data assignment matrix} $\allocmat \in \{0, 1\}^{\nworker \times \ngrad}$ where $\allocmat[{\wind,\gind}]$ is equal to $1$ if the $\gind$-th data sample is given to the $\wind$-th worker and $0$ otherwise,
    \item a set of \textbf{encoding functions} used by the workers
    \mbox{$\encfunset \defeq \left\{
        \encfun[\encind] \colon \galpha^{\graddim \times \ngrad} \to \galpha^{\respdim_{\encind,1}} \times \cdots \times \galpha^{\respdim_{\encind,\nworker}}
        \mid \encind \in \range{\nencfun}
    \right\}$},
    \item a \textbf{multi-round protocol} $\proto = (\proto_1,\proto_2)$ in which $\proto_1$ selects the indices of the encoding functions to be used by the workers based on all previous responses at the start of each iteration, and $\proto_2$ selects gradients to be locally computed at the main node based on all previous responses at the end of each iteration,
    \item and a \textbf{decoding function} $\decfun$ used by the main node after running the protocol $\proto$ that outputs the correct full gradient whenever there are at most $s$ malicious workers.
\end{itemize}
\end{definition}

We suppose that each worker transmits an \emph{initial response}, which is a codeword symbol of an erasure-correcting gradient code~\cite{tandon}, before the protocol starts.
In each round $\roundind = 1,\dots,\nround$ of the interactive protocol, worker $\wind$ computes its responses based on the partial gradients assigned according to the assignment matrix. A response can be written as
$
    \iresp_{\wind} \defeq \left[ \encfun[\encind] \left( \G \right) \right]_\wind \in \galpha^{\respdim_{\encind,\wind}}
$
when $\encfun[\encind]$ is asked for, where we define $\G = \left(\pgrad[1],\pgrad[2],\dots,\pgrad[\ngrad]\right) \in \galpha^{\graddim \times \ngrad}$.
The main node receives
$\noisyiresp_{\wind} = \iresp_{\wind} + \err_{\wind}$,
where $\err_{\wind}=\zero$ for honest workers, and after which it may compute some partial gradients locally.

\begin{definition}[Figures of merit~\cite{hofmeisterTradingCommunication2023}]
    We evaluate an \BGC scheme by the following measures:
    \begin{itemize}[itemsep = 0em]
        \item
        The number of \textbf{local computations} $\localcomp$ of partial gradients at the main node.
        \item
          The \textbf{replication factor} $\replfact$, which is the average number of workers to which each sample is assigned, i.e.,
        \begin{equation*}
        \replfact \defeq \frac{\sum_{\wind \in \range{\nworker},\gind \in \range{\ngrad}} \allocmat[{\wind, \gind}]}{\ngrad}.
        \end{equation*}
        \item
        The \textbf{communication overhead}\footnote{We slightly adapt the definition of the communication overhead from~\cite{hofmeisterTradingCommunication2023} by taking the gradients and consequently the symbols we count to be from a (potentially infinite) field rather than a finite alphabet.} $\commoh$, which is the maximum number of symbols from $\galpha$ transmitted from the workers to the main node during $\proto$, i.e.,
        \begin{equation*}
            \commoh \defeq \sum_{\roundind \in \range{\nround}, \wind \in \range{\nworker}} \respdim_{\encind_\roundind,\wind}.
        \end{equation*}
    \end{itemize}
\end{definition}
We restrict this work to \emph{regular data assignments} $\assignmentset$, where each partial gradient is replicated the same number of times and each worker is allocated at least one partial gradient to compute. Formally, we have
\mbox{\(
    \assignmentset = \left\{ \allocmat \in \{0, 1\}^{\nworker \times \ngrad} \mid \sum_{\wind=1}^{\nworker} \allocmat[{\wind, \gind}] = \replfact, \sum_{\gind=1}^{\ngrad} \allocmat[{\wind, \gind}] \geq 1  \right\}
\)}. Since we focus on worst-case guarantees, this does not limit the scope of this paper as we can suppose an adversary introduces errors in partial gradients with the smallest replication. 

%% file: setting_example.tex
\definecolor{honestcolor}{RGB}{76,193,46}
\definecolor{maliciouscolor}{RGB}{192,36,105}
\begin{tikzpicture}[
        node distance=1.6cm,
        master/.style={circle,minimum width=1cm,minimum height=1cm,text centered,text width=1cm,draw=black,fill=honestcolor!50},
        honest/.style={rectangle,minimum width=1.5cm,minimum height=1cm,text centered,text width=1cm,draw=black,fill=honestcolor!50},
        malicious/.style={starburst,text centered,text width=0.75cm,draw=black,fill=maliciouscolor!50}
    ]
    \node (master) [master] {Main Node};
    \node (W2) [honest, below=1cm of master] {$W_2$\\$x_2, x_3$};
    \node (W1) [honest, left=1cm of W2] {$W_1$\\$x_1, x_2$};
    \node (W3) [malicious, right=1cm of W2] {$W_3$\\$x_1, x_3$};
    \draw [stealth-] (master) edge[bend right=30] node[fill=white, anchor = east, xshift=0.2cm](l1){$\mathbf{g}_1/2 + \mathbf{g}_2$}(W1);
    \draw [stealth-] (master) edge node[fill=white, anchor=north, yshift=0.26cm](l2){$\mathbf{g}_2 - \mathbf{g}_3$}(W2);
    \draw [stealth-] (master) edge[bend left=30] node[fill=white, anchor=west, xshift=-0.35cm, yshift=-0.2cm](l3){$\mathbf{g}_1/2 + \mathbf{g}_3 + {\color{maliciouscolor}\mathbf{e}}$}(W3);
\end{tikzpicture}

%% file: scheme_construction.tex
In this section, we present an \BGC scheme that is capable of identifying malicious workers given a replication factor $\replfact=\nmalicious+\nhonest$ for an integer $\nhonest \geq 1$ and an assignment matrix $\allocmat \in \assignmentset$.
For notational convenience, we define $\rth = \nworker - (\nmalicious + \nhonest)$. 
The scheme works by successively decreasing the number of non-eliminated malicious workers $\nmalicious_\roundind$ until the full gradient $\tgrad$ can either be decoded through an error correcting code (ECC) or there are no more disagreements about the value of $\tgrad$ among the workers.
The steps of the scheme are outlined in \cref{flowchart}. 
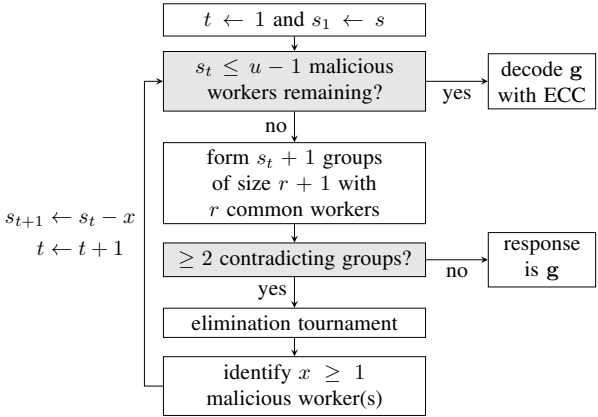
\begin{figure}[h!]
  \centering
  \resizebox{0.9\linewidth}{!}{\input{flowchart}}
  \caption{Flowchart illustrating the steps of the interactive protocol.}
  \label{flowchart}
\end{figure}

We explain the scheme for the setting in \cref{example} before we give the general description.
In \cref{example}, $\nworker = 3$, $\nmalicious = 1$, $\replfact = 2$, and $\nhonest = 1$. The algorithm starts with $\nmalicious_\roundind = \nmalicious > \nhonest - 1$, i.e., there is one unidentified malicious worker, and the correct gradient cannot be directly decoded.
Initially, each worker sends a linear combination of the assigned partial gradients. The main node forms $\nmalicious_\roundind + 1 = 2$ groups of size $\rth+1 =2$ workers with $\rth = 1$ worker in common, i.e., $\groupset[1]=\{W_1,W_3\}$ and $\groupset[2]=\{W_2,W_3\}$. The decoded full gradient from $\groupset[1]$ is $\sum_{i=1}^{3}\mathbf{g}_i + \mathbf{e}$ and from $\groupset[2]$ is $\sum_{i=1}^{3}\mathbf{g}_i + 2\mathbf{e}$, which are the same if and only if $\mathbf{e} = \zero$. If the decoded full gradient from both groups differs, the main node applies a binary search to find a contradicting partial gradient value as follows.
The main node asks each worker to encode $\sum_{i=1}^{2}\mathbf{g}_i$ so that it can compare the decoded value from $\groupset[1]$ and $\groupset[2]$ again. In this case, each worker sends the part of their linear combination corresponding to $\mathbf{g}_1$ and $\mathbf{g}_2$. Thus, $W_1$ sends $\mathbf{g}_1/2 + \mathbf{g}_2$ and $W_2$ sends $\mathbf{g}_2$. Suppose $W_3$ sends $\mathbf{g}_1/2 + \mathbf{e}$. The group responses are then decoded
as $\sum_{i=1}^{2}\mathbf{g}_i + \mathbf{e}$ for $\groupset[1]$ and as $\sum_{i=1}^{2}\mathbf{g}_i + 2\mathbf{e}$ for $\groupset[2]$. Observe that the response from $W_3$ is weighted differently in the decoding steps from $\groupset[1]$ and $\groupset[2]$ respectively; therefore, the adversary cannot induce the same error in both groups. Similarly, the main node asks the groups to encode $\mathbf{g}_2$.
Suppose that $W_3$ answers honestly; therefore, both groups yield the same value.
Since each group has already committed to a value for $\sum_{i=1}^{2}\mathbf{g}_i$ before, the main node can deduce the claimed value for $\mathbf{g}_1$ as $\mathbf{g}_1 + \mathbf{e}$ for $\groupset[1]$ and $\mathbf{g}_1 + 2\mathbf{e}$ for $\groupset[2]$. As only $W_3$ was assigned $\mathbf{g}_1$ in $\groupset[2]$, this means that $W_3$ claims $\mathbf{g}_1 + 2\mathbf{e}$ as value for $\mathbf{g}_1$.
Hence, by locally computing $\mathbf{g}_1$ the main node identifies $W_3$ as malicious and decodes the full gradient correctly.

We summarize the capabilities of our scheme in the following theorem. In summary, our scheme generalizes ideas used in~\cite{draco} for a cyclic repetition data assignment and $\nhonest=\nmalicious+1$, to any regular data assignments with $\replfact = \nmalicious + \nhonest$ and it is also applicable for $\ngrad \neq \nworker$.
\begin{theorem}
\label{theorem: schemeoverview}
    The scheme constructed below is an \BGC scheme with a parameter $\nhonest$, $1 \leq \nhonest \leq \nmalicious + 1$, and requires $\localcomp \leq \nmalicious + 1 - \nhonest$ local computations, a replication $\replfact = \nmalicious + \nhonest$ and a communication overhead $\commoh \leq (\rth+2)(\nmalicious + 1 - \nhonest)\lceil \log_2(\ngrad)\rceil $ for any data assignment $\allocmat \in \assignmentset$.
\end{theorem}

We now elaborate on the protocol in \cref{flowchart}. A key idea is the construction of a set of encoding and decoding matrices ${\W}^{(\mathbf{a})} \in \galpha^{\ngrad \times \nworker}$ and $\B^{(\roundind)} \in \galpha^{\nworker \times \ngroup}$ that allow the main node to obtain any desired linear combination of the partial gradients from the workers' computations. We explain the construction of the encoding and decoding matrices in \cref{sec: encoding-matrix}. We then give a detailed description of the individual components of our scheme.

\subsection{Construction of the Encoding and Decoding Matrices}
\label{sec: encoding-matrix}

We construct a set of encoding and decoding matrices, ${\W}^{(\mathbf{a})}$ and $\B^{(\roundind)}$, of a gradient code such that the main node can recover \emph{any sum of partial gradients} from the responses of any group of workers $\groupset[]$ of size $\rth+1$. The workers can only compute assigned partial gradients according to $\allocmat$. 

Let $\MDS \in \galpha^{(\rth+1) \times \nworker}$ be the generator matrix of an $(\nworker,\rth+1)$ MDS code. We choose a Reed-Solomon code; that is, $\MDS$ is a Vandermonde matrix with distinct evaluation points. Note that this construction exists whether $\galpha$ is a finite field, $\R$, or $\C$. For more details on a construction over $\C$, see~\cite{draco}.
Let $\mathcal{I}_\gind = \{\wind \mid \allocmat[{\wind, \gind}] = 0\}$ be the set of row indices of zero entries in $\allocmat[{\cdot,\gind}]$, that is, the set of workers that are not assigned $\pgrad$. Let $\mathbf{a}=(a_1,\dots,a_\ngrad)^\tp \in \galpha^\ngrad$ denote the coefficients of a desired linear combination of the partial gradients $\sum_{i=1}^{\ngrad} a_i \tgrad_i$. Let $(\cdot|\cdot)$ be the concatenation operator, we define ${\W}^{(\mathbf{a})} = \left( {\mathbf{Q}}^{(\mathbf{a})} \mid \mathbf{a} \right) \MDS$ where ${\mathbf{Q}}^{(\mathbf{a})}=({\mathbf{q}}^{(\mathbf{a})}_1,\dots,{\mathbf{q}}^{(\mathbf{a})}_\ngrad)^\tp \in \galpha^{\ngrad \times \rth}$ is the solutions of the system of equations
\begin{equation}
    \label{eq:q}
    \zero = \left( {\mathbf{q}^{(\mathbf{a})}_\gind}^\tp \mid a_\gind \right) \MDS_{\cdot, \mathcal{I}_\gind}, \quad \forall i \in [\ngrad].
\end{equation}
Since $\MDS$ generates an MDS code of dimension $\rth+1$ and $\card{\mathcal{I}_\gind} = \rth$, we can obtain $\ngrad$ vectors $\mathbf{q}^{(\mathbf{a})}_\gind \in \galpha^\rth$ uniquely. 
By this, we establish a result analogous to~\cite[Lemma 3]{draco} showing that ${\W}^{(\mathbf{a})}$ satisfies the desired properties. 
\begin{lemma}
    \label{lemma:encoding-matrix}
    For all $\wind,\gind$, if $\allocmat[{\wind, \gind}] = 0$, then  ${\W}_{\gind, \wind}^{(\mathbf{a})} = 0$. Furthermore, for any index set $\groupset[]$ such that $\card{\groupset[]} \geq \rth + 1$, the column span of $\W_{\cdot, \groupset[]}^{(\mathbf{a})}$ contains the vector $\mathbf{a}$.
\end{lemma}
\begin{proof}
    Suppose $\allocmat[{\wind, \gind}] = 0$ for some $\wind,\gind$. Then $\wind \in \mathcal{I}_\gind$, and by construction ${\W}_{\gind,\wind}^{(\mathbf{a})} = \left( {\mathbf{q}^{(\mathbf{a})}_\gind}^\tp \mid a_i \right) \MDS_{\cdot, \wind} = 0$.
    For any index set $\groupset$ with $\card{\groupset} = \rth+1$, because of the Vandermonde structure of $\MDS$ we can show that there exists a vector
    \begin{equation}
        \label{eq:combvec}
        \combvec = \MDS_{\cdot, {\groupset}}^{-1} \left(0, \dots, 0, 1\right)^\tp,
    \end{equation}
    and
    \(
        {\W}_{\cdot,\groupset}^{(\mathbf{a})} \combvec = \left( {\mathbf{Q}}^{(\mathbf{a})} \mid \mathbf{a} \right) \MDS_{\cdot,\groupset} \MDS_{\cdot, {\groupset}}^{-1} \left(0, \dots, 0, 1\right)^\tp = \mathbf{a}.
    \)
\end{proof}
A closed form expression for $\combvec[{\groupset[]}]$ when $\MDS$ generates a Reed-Solomon code is given in \cref{sec: b-closedform}.
In round $\roundind$, the main node constructs the decoding matrix \mbox{$\B^{(\roundind)} \defeq \left(\combvec[1], \combvec[2], \dots, \combvec[\ngroup] \right) \in \galpha^{\nworker \times \ngroup}$} for a choice of $\ngroup$ different groups $\groupset$, $\groupind \in [\ngroup]$. The decoding matrix is independent of $\mathbf{a}$.
This matrix contains the combining coefficients for all groups as columns and satisfies ${\W}^{(\mathbf{a})}\B^{(\roundind)} = [\mathbf{a}, \dots, \mathbf{a}] \in \galpha^{\ngrad \times \ngroup}$.

\subsection{Workers' Response and Contradicting Groups}

\emph{Workers' response: }When the main node wishes to compute $\G\mathbf{a}$, each worker responds with a vector
\mbox{$\iresp_{\wind}^{(\mathbf{a})} = \G \W^{(\mathbf{a})}_{\cdot,\wind} \in \galpha^{\graddim}$}. 
We define the workers' responses by the matrix \mbox{$\Z^{(\mathbf{a})} = ( \iresp_{\wind}^{(\mathbf{a})} ) = \G{\W}^{(\mathbf{a})} \in \galpha^{\graddim \times \nworker}$}
and accordingly write \mbox{$\noisyZ^{(\mathbf{a})}=\Z^{(\mathbf{a})} + \E$}, where the error matrix $\E \defeq (\err_{\wind} ) \in \galpha^{\graddim \times \nworker}$ includes the error terms imposed by each worker as columns.

We define the encoding matrix used in the initial response as $\W = \W^{(\one)}$
and let $\Z = \Z^{(\one)}$ be the {initial} responses of the workers dictated by $\W$. From~\cref{lemma:encoding-matrix}, the full gradient can then be decoded from any group of $\rth + 1$ workers\footnote{We chose $\rth = \nworker - \replfact = \nworker - (\nmalicious+\nhonest)$ since if $\rth < \nworker - \replfact$, then there exists a group selection which does not include any workers which have computed $\pgrad$ for some $\gind$ and the full gradient cannot be recovered.}. In our protocol, we only require the reconstruction of either the full gradient or the sum of a subset of partial gradients, i.e., $\mathbf{a} \in \{0,1\}^\ngrad$, from any group $\groupset$ of workers. 
\begin{remark}
    \label{rem: encoding-matrix}
    For $\mathbf{a} \in \{0,1\}^\ngrad$, $\W^{(\mathbf{a})}$ differs from $\W$ only by setting certain rows of $\W$ to be zero rows, corresponding to the partial gradients which are not asked for. Thus, each worker $\wind$ at any stage of the protocol transmits $\G{\W}_{\cdot,\wind}^{(\mathbf{a})}$, which as a linear combination of \pgrad differs from $\G{\W}_{\cdot,\wind}$ only by setting some coefficients of the linear combination to $0$.
\end{remark}

\emph{Contradicting groups: }Let $\nmalicious_t$ be the number of potential malicious workers present at round $\roundind$ with $\nmalicious_1 = \nmalicious$. To decode the correct full gradient and detect the presence of malicious workers, the main node creates $m = \nmalicious_\roundind+1$ groups of workers, each of cardinality $\rth+1$ and that all have $\rth$ workers in common. An exhaustive search over all $\binom{\nworker}{\rth+1}$ groups is computationally prohibitive. We show that only $\nmalicious_\roundind+1$ groups constructed using our strategy are enough and show later that any choice of less than $\nmalicious_\roundind+1$ groups is insufficient.

Let $\wind_1\dots,\wind_{\rth+\nmalicious_\roundind+1}$ be any indexing of $\rth+\nmalicious_\roundind+1$ non-eliminated workers. Then we define the groups as
\begin{equation}
    \label{eq:grouping_scheme}
    \groupset = \{ \wind_1,\dots,\wind_\rth \} \cup \{ \wind_{\rth+\groupind} \},
    \quad
    \groupind = 1,\dots,\nmalicious_\roundind+1.
\end{equation}

From each group $\groupind \in \range{\nmalicious_\roundind+1}$, the main node computes the claimed value of the desired linear combination of the partial gradients denoted by a group's response $\gestim^{(\groupind,\mathbf{a})}$, i.e.,
\begin{equation}
    \label{eq:decode-group-value}
    \gestim^{(\groupind,\mathbf{a})} \defeq  \noisyZ^{(\mathbf{a})} \combvec = \G \mathbf{a} + \E \combvec.
\end{equation}
Groups that yield different values of $\gestim^{(\groupind,\mathbf{a})}$ are called \emph{contradicting groups}.
If there are no contradicting groups, \cref{lemma:achievable-group-number} shows that $\gestim^{(\groupind,\mathbf{a})}$ is the desired linear combination of the partial gradients. In particular, if there were no contradicting groups at round $\roundind = 1$, the main node decodes the full gradient correctly.

\begin{lemma}
    \label{lemma:achievable-group-number}
    Let $\nmalicious_\roundind$ be the number of unidentified malicious workers and consider $\nmalicious_\roundind+1$ groups constructed as in \cref{eq:grouping_scheme}. If all $\nmalicious_\roundind+1$ groups yield the same response, it must be correct.
    \end{lemma}
\begin{proof}
We only outline the proof idea here and defer the detailed proof to Appendix \ref{sec : lemma3-proof}.

Let $\E$ be the error matrix, and $\B^{(\roundind)}$ be the decoding matrix for the chosen $\nmalicious_\roundind+1$ groups. We show that the linear system of equations $\E_{\cdot,\mathcal{S}} \B_{\mathcal{S},\cdot}^{(\roundind)} = \one_{d \times (\nmalicious_\roundind + 1)}$ has no solution, with $\mathcal{S} \subset \range{\nworker}$ being the indices
of the malicious workers and $\card{\mathcal{S}} = \nmalicious_\roundind$. This would imply that for any arbitrary set of up to $\nmalicious_\roundind$ malicious workers, the adversary cannot corrupt their responses (by adding error vectors to their outputs) such that it induces each group to return the same incorrect full gradient upon being decoded by the main node. We derive a closed-form expression for $\B^{(\roundind)}$ from \cref{eq:combvec} by using the known inverse of a Vandermonde matrix~\cite{Rawashdeh2019ASM}, thus showing no solution exists.
Hence, choosing any $\nmalicious_\roundind+1$ groups, it is guaranteed that they cannot yield the same claimed full gradient value for any non-zero error matrix $\E_{\cdot,\mathcal{S}}$.
\end{proof}

If there are at least two contradicting groups, the main node needs to identify one group without malicious workers to decode the correct full gradient from. To detect which workers are malicious, the main node chooses two contradicting groups and plays the elimination tournament (explained next) to detect at least one malicious worker.

Afterward, the main node decreases $\nmalicious_\roundind$ accordingly and proceeds to the next round $\roundind+1$. At this round, the main node forms $\nmalicious_{\roundind+1}+1$ groups and repeats this process until only $\nhonest-1$ malicious workers remain undetected. The main node can then decode the full gradient correctly using the error-correction capability of the designed encoding and decoding matrices, cf.~\cref{cor:apply_DRACO}. Limiting the number of required rounds to $\nmalicious-(\nhonest-1)$ shows how the additional redundancy $\nhonest > 1$ can be leveraged to reduce the communication overhead.

\begin{lemma}
    \label{cor:apply_DRACO}
    If and only if there are at most $\nhonest-1$ unidentified malicious workers in round $\roundind$, then the main node can decode the full gradient without any additional local gradient computations or communication from the initial responses.
    \begin{proof}
        \newcommand{\unident}{\ensuremath{{\bar{\mathcal{S}}}}}
        Note that effectively, the main node tries to decode the information word $(\mathbf{c}_1,\dots,\mathbf{c}_{\rth+1}) = \G \left( \mathbf{Q}^{(\one)} \mid \one \right)$ encoded as $(\iresp_{1},\dots,\iresp_{\nworker})$ by the MDS code generated by $\MDS$. Let $\unident \subset \range{\nworker}$ be the indices of already identified malicious workers in round~$\roundind$. If there are only $\nhonest-1$ unidentified malicious workers, i.e., $\card{\unident}=\nmalicious-(\nhonest-1)$, then the messages from the remaining workers $\Z_{\cdot, \range{\nworker} \setminus \unident}$ can be viewed as a codeword from an MDS code generated by the matrix $\MDS_{\cdot, \range{\nworker} \setminus \unident}$. This is equivalent to puncturing the original code generated by $\MDS$, and yields an $(\nworker-\card{\unident},\rth+1) = (\nworker-\nmalicious+(\nhonest-1),\nworker-\nmalicious-(\nhonest-1)$ MDS code. This code has minimum distance $2(\nhonest-1)$, and therefore, can correct the error of the $\nhonest-1$ unidentified malicious workers.
        The converse follows from the symmetrization attack described in~\cite{draco}.
    \end{proof}
\end{lemma}

\subsection{Elimination Tournament to Identify Malicious Workers}
\label{sec:elimination-tournament}

During an elimination tournament, the main node runs matches between two contradicting groups of workers, say $\groupset[{\groupind_1}]$ and $\groupset[{\groupind_2}]$, to identify at least one malicious worker introducing errors in its computation. The elimination tournament is a generalization of the elimination tournament presented in ~\cite{hofmeisterTradingCommunication2023}. The tournament aims to find \emph{one} partial gradient on whose value at least two workers disagree. This is guaranteed since the groups are contradicting, i.e., $\gestim^{(k_1,\mathbf{a})}$ and $\gestim^{(k_2,\mathbf{a})}$ have different values. By computing this partial gradient locally, the main node identifies at least one malicious worker. 

The main node constructs a full binary tree, called the \emph{match tree} labeled by partial gradients. The root node is labeled by the sum of all partial gradients. Each node has two children: the first is labeled by the sum over the first half of the parent's partial gradients, and the second by the sum over the second half. Proceeding recursively, each leaf node is labeled by an individual partial gradient. The contradicting groups claim different values for the root node. The matches walk on the tree to reach a contradicting leaf. At every match, the main node asks the workers of both groups to send a response from which the main node can decode one of two child nodes of a parent node whose label is known.
This is always possible by \cref{rem: encoding-matrix}.
Further, we do not need the workers or main node to compute a new encoding matrix from scratch at any stage of the match tree. The main node can also infer the groups' claims for the other child node. The two groups must disagree on one of these nodes. Proceeding recursively, the main node can always reach a leaf node where the groups differ. The main node computes the value of the leaf locally and detects which workers are behaving maliciously.

\begin{lemma}
    \label{rem: partialgradient}
    Given two groups of workers $\groupset[{\groupind_1}]$ and $\groupset[{\groupind_2}]$ with $\gestim^{(\groupind_1, \one)} \neq \gestim^{(\groupind_2, \one)}$ an elimination match always reaches a leaf of the match tree, for which not all workers claim the same label. Accordingly, the main node can identify at least one malicious worker after each local computation.
\end{lemma}

%% file: flowchart.tex
\begin{tikzpicture}[
    node distance=1.3cm,
    flowchart/.style={rectangle,text centered,text width=4cm,draw=black},
    flowchartdecision/.style={rectangle,fill=black!10!white,text centered,text width=4cm,draw=black},
    flowchart2/.style={rectangle,text centered,text width=1.5cm,draw=black}
]
    \node (n1) [flowchartdecision] {$\nmalicious_\roundind \leq \nhonest-1$ malicious workers remaining?};
    \node (start) [flowchart, node distance=1cm, above of=n1] 
        {$\roundind \gets 1$ and $\nmalicious_1 \gets \nmalicious$};
    \node (o1) [flowchart2, right=1cm of n1] {decode $\tgrad$ with ECC};
    \node (n2)  [flowchart, below=0.5cm of n1] {form $\nmalicious_\roundind+1$ groups of size $\rth + 1$ with $\rth$ common workers};
    \node (n3)  [flowchartdecision, below=0.3cm of n2] {$\geq 2$ contradicting groups?};
    \node (o2)  [flowchart2, right=1cm of n3] {response is $\tgrad$};
    \node (n4)  [flowchart, below=0.5cm of n3] {elimination tournament};
    \node (n5)  [flowchart, below=0.3cm of n4] {identify $x \geq 1$ malicious worker(s)};
    \draw [stealth-] (n5) -- node{}(n4);
    \draw [stealth-] (n4) -- node[anchor=east]{yes}(n3);
    \draw [stealth-] (n3) -- node{}(n2);
    \draw [stealth-] (n2) -- node[anchor=east]{no}(n1);
    \draw [stealth-] (o1) -- node[anchor=north]{yes}(n1);
    \draw [stealth-] (o2) -- node[anchor=north]{no}(n3);
    \draw [stealth-] (n1) -- node{}(start);
    \draw [stealth-] 
      (n1.west) -- +(-0.3,0) |- node[pos=0.25, xshift=-1.2cm] {$\begin{aligned} s_{t+1} &\gets s_{t}-x \\ t &\gets t+1 \end{aligned}$} (n5);
  \end{tikzpicture}

%% file: analysis.tex
In the following, we show that our proposed grouping strategy is optimal (in terms of the number of group comparisons) and discuss the figures of merit of our scheme.

\subsection{Optimality of the Grouping Strategy}
\label{sec: groupisoptimal}

\begin{theorem}
    Consider an \BGC scheme with $\replfact = \nmalicious + \nhonest$ with $1 \leq \nhonest \leq \nmalicious + 1$. For every selection of $m \leq \nmalicious_\roundind$ groups of workers, each of size at least $ \nworker - ( \nmalicious + \nhonest ) + 1 $, it is possible that all the groups return an identical erroneous response.
\end{theorem}
\begin{proof}
    Since the main node is unaware of the malicious workers' identities, the worst-case group selection results in each group having at least one malicious worker. It is then sufficient to show that given $\B$ as defined above, there exists a non-zero error matrix $\E$ with at most $\nmalicious$ non-zero columns such that $\E\B = \one \in \galpha^{\graddim\times \nmalicious}$ (the situation in which all groups return an identical corrupted response, see \cref{eq:decode-group-value}). For simplicity, we consider $\graddim = 1$ and drop superscripts $\roundind$ and $\mathbf{a}$. The generalization to an arbitrary value of $\graddim$ is straightforward.

    Observe the rank of $\B$ is at most $\nmalicious$. Since  $\Z\B = g\one \in \galpha^{1 \times \nmalicious}$, the all-one vector must be in the row span of $\B$. Since only at most $\nmalicious$ rows of $\B$ can be linearly independent, there exists a set of $\nmalicious$ rows of $\B$ whose span contains the all-one vector. Hence, there exists a non-zero matrix $\E$  with at most $\nmalicious$ non-zero columns such that $\E\B = g\one \in \galpha^{1 \times \nmalicious}$.
\end{proof}

From the above theorem, we see that our proposed grouping strategy is optimal in terms of the number of groups ($\nmalicious + 1$) that we formed for group comparisons.

\subsection{Figures of Merit}
\label{sec: Communication Overhead}

Recall that as long as the main node does not compute all partial gradients locally, i.e., $\localcomp < \ngrad $, the replication factor of any \BGC scheme is lower-bounded as $\replfact \geq \nmalicious + 1$. Moreover, if $\localcomp = 0$, then the replication factor of any \BGC scheme is lower-bounded as $\replfact \geq 2\nmalicious + 1$~\cite{draco}. Conversely, if \mbox{$\replfact \geq 2s + 1$}, then
$\localcomp = \commoh = 0$ is achievable using error-correcting codes, which is also reflected in our scheme. For \mbox{$\nmalicious + 1 \leq \replfact \leq 2\nmalicious + 1$}, the scheme's capabilities are summarized in \cref{theorem: schemeoverview}.

Our analysis disregards the communication cost from the main node to the workers. The amount of information communicated from the main node to the workers is negligible since, at each match of an elimination tournament, the main node communicates only one bit, indicating the descent direction in the match tree.
Similar to the fractional repetition scheme and the fractional repetition scheme in~\cite{hofmeisterTradingCommunication2023}, the resulting communication overhead in the presented scheme is logarithmic in the number of partial gradients $\ngrad$, which is (practically) much larger than the number of workers.
We remark that for $\nhonest=1$, the fractional repetition scheme is a special case of the scheme presented in this paper where the group size is $1$.
For $\nhonest > 1$, the number of local computations required by the presented scheme is higher than the fractional repetition scheme.
In turn, the presented scheme applies to a much larger class of data assignments, namely the set of regular data assignments $\assignmentset$.
An open question is whether the number of local computations can be reduced similarly for all regular data assignments.
Future research includes the analysis of fundamental limits for general data assignments, improvements (if possible) of the local computation load by leveraging additional redundancy, and the search for other efficient \BGC schemes.

%% file: appendix.tex
\subsection{Proof of \cref{theorem: schemeoverview}}

By its construction in \cref{sec:scheme}, the scheme is guaranteed to eliminate at least $1$ malicious worker per round. Therefore, it terminates after a finite number of rounds. If the main node still finds contradicting groups, it outputs the correct full gradient $\tgrad$ by decoding over an error-correcting code, cf. \cref{cor:apply_DRACO}. Otherwise, the main node finds the correct full gradient directly from the groups' responses, cf. \cref{lemma:achievable-group-number}. Hence, the presented scheme is a valid \BGC scheme.

A local computation is only run once per round when the main node finds two contradicting groups (without any local computations). As long as $\nmalicious_\roundind > 0$, our grouping strategy is guaranteed to find two contradicting groups. Given two contradicting groups, the elimination tournament always finds contradicting claimed values for one partial gradient. Since this procedure identifies one malicious worker in each round, after $\nmalicious + 1 - \nhonest$ rounds, the correct full gradient can be obtained without any more local gradient computations according to \cref{cor:apply_DRACO}. Thus, at most $\localcomp \leq \nmalicious + 1 - \nhonest$.

Finally, we bound the communication overhead, which is due to the elimination tournament played in each round. Having identified two contradicting groups with group responses $\gestim^{(\groupind_1,\one)} \neq \gestim^{(\groupind_2,\one)}$, the main node can run the elimination tournament only on a single coordinate for which the group responses differ. Thus, at each level of the match tree, each competing worker transmits one symbol from $\galpha$. 
The number of competing workers is $\card{ \groupset[{\groupind_1}] \cup \groupset[{\groupind_2}]} = \rth + 2$.
The number of levels corresponds to the height of the match tree (excluding the root node), which is $\lceil \log_2(\ngrad)\rceil$. So $(\rth+2)\lceil \log_2(\ngrad)\rceil$ symbols are transmitted from the workers to the main node each round. Since there are at most $\nmalicious + 1 - \nhonest$ rounds, we have $\commoh \leq (\rth+2)(\nmalicious + 1 - \nhonest)\lceil \log_2(\ngrad)\rceil$.
     
\subsection{Proof of \cref{lemma:achievable-group-number}}
\label{sec : lemma3-proof}

Assume the presence of at most $\nmalicious$ malicious workers. Consider a selection of $\nmalicious_\roundind + 1$ groups as in the grouping scheme in \cref{eq:grouping_scheme}. For simplicity, we drop the round index $\roundind$ in the sequel.
For the sake of presentation, we prove the statement for scalar partial gradients $\pgrad$, i.e., $\graddim = 1$. For $\graddim > 1$, the same line of arguments applies for each of the $\graddim$ gradient dimensions, and the proof generalizes straightforwardly.

\newcommand{\mset}{\ensuremath{{\mathcal{E}}}}
Denote the index set of all malicious workers (this set is unknown to the main node) as $\mset$.
According to \cref{eq:decode-group-value}, the main node obtains the response
\begin{equation*}
    \gestim^{(\groupind,\one)} = \G \one + \E \combvec = \tgrad + \E \combvec
\end{equation*}
from the initial response of group $\groupset$, $\groupind \in \range{\nmalicious+1}$.
Thus, each group's response is decoded as the correct full gradient plus an error term. To prove the lemma, it suffices to show that the error terms from at least two groups differ or that all groups' error terms are zero. If all the error terms are zero, the main node immediately obtains the correct full gradient. Therefore, we suppose there is a non-zero error term for the remainder of the proof.

We jointly write all groups' error terms using the decoding matrix\footnote{Although defined as a matrix in general, $\E$ denotes a vector here since we assume $\graddim=1$. We abuse notation here for the sake of clarity.} $\B$ as $(\E \combvec[1],\dots,\E \combvec[\nmalicious+1]) = \E \B = \E_{\cdot,\mset} \B_{\mset,\cdot}$.
Showing that at least two error terms differ is now equivalent to showing that the equation system given by $\E_{\cdot,\mset} \B_{\mset,\cdot} = \lambda \one_{1 \times (\nmalicious+1)}$ is inconsistent (has no solution) for any choice of $\mset$, where $\lambda$ is some non-zero scalar. For ease of presentation, we work with $\B^\tp$ and $\E^\tp$ (so the columns of $\B^\tp$ correspond to the workers and the rows of $\E^\tp$ correspond to the respective errors they induce). Thus, we consider the transposed system of equations $\B_{\mset,\cdot}^\tp \E_{\cdot,\mset}^\tp = \lambda \one_{(\nmalicious+1) \times 1}$.%
\footnote{Note that $\B_{\mset,\cdot}^\tp$ refers to, first, restricting the matrix $\B$ to the rows indexed by $\mset$, and afterwards, transposing the restricted matrix.}
Showing that this system of linear equations is inconsistent is equivalent to showing that there exists a pivot in the last column of the echelon form of the augmented matrix (refer \cref{sec: pivot} or \cite{Treil_2017}). 

Now consider the $(\nmalicious+1)\times(\nmalicious+1)$ augmented matrix $(\B_{\mset,\cdot}^\tp \mid \one_{(\nmalicious+1) \times 1})$ and show that it contains a pivot in the last column through row and column operations. Out of the $\nmalicious$ malicious workers, assume $k$ are in the root group $\mathcal{R}=\{j_1,\dots,j_\rth\}$, cf. \cref{eq:grouping_scheme}. Denote these workers $E_1, E_2, .... , E_k$ and the remaining $\nmalicious-k$ workers as $E_{k+1}, ... , E_\nmalicious$. Observe that the column of $\B^\tp$ corresponding to worker $E_i$, $i>k$ (who is not in the root group) will only contain one non-zero entry corresponding to the group the respective worker is in. Thus, we can rearrange the columns of $(\B_{\mset,\cdot}^\tp \mid \one_{(\nmalicious+1) \times 1})$ so that the first $\nmalicious-k$ columns correspond to malicious workers not in the root group. Then, we can rearrange the first $\nmalicious-k$ rows in a way such that the submatrix of $(\B_{\mset,\cdot}^\tp \mid \one_{(\nmalicious+1) \times 1})$ given by the first $\nmalicious-k$ rows and columns is a diagonal matrix, and the submatrix of $(\B_{\mset,\cdot}^\tp \mid \one_{(\nmalicious+1) \times 1})$ given by the last $k+1$ rows and first $\nmalicious-k$ columns is a zero matrix. Swapping rows and columns rearranges the equations and does not change the equation system. Without loss of generality, assume that the last $k$ columns of the rearranged $\B_{\mset,\cdot}^\tp$ correspond to workers $E_1, E_2, .... , E_k$ and the first $\nmalicious-k$ columns of the rearranged $\B_{\mset,\cdot}^\tp$ correspond to workers $E_{k+1}, E_{k+2}, .... , E_\nmalicious$. Now, to show a pivot in the last column of the augmented matrix exists, it suffices to show that the determinant of the submatrix of the rearranged $(\B_{\mset,\cdot}^\tp \mid \one_{(\nmalicious+1) \times 1})$ given by the last $k+1$ rows and columns is non-zero (since this would imply that there exists a set of row operations which would change this sub-matrix to the identity matrix). Next, we find the entries of this determinant. We already know the entries in the last column by definition.

To find the remaining entries, we explicitly determine $\combvec$. Recall that $\combvec$ is given by the last column of $\MDS_{\cdot, {\groupset}}^{-1}$.
Notice that due to the Vandermonde structure of $\MDS$, each worker $\wind$ is associated with a distinct evaluation point $\omega_\wind \in \galpha$.
The closed-form expression for the inverse of a Vandermonde matrix is known and given in~\cite{Rawashdeh2019ASM}. Hence, for a group $\groupset$ of $\rth+1$ workers, we find the entry of $\combvec$ corresponding to worker $\wind \in \groupset$, which is given by $\left( \prod_{\ell \in \groupset \setminus \{\wind\}}(\omega_\wind - \omega_\ell) \right)^{-1}$. Refer to \cref{sec: b-closedform} for the detailed derivation.

\begin{figure*}[t]
\begin{equation}
\label{eq:reduced_matrix}
\begin{pmatrix}
 \frac{1}{\omega_{E_1} - \omega_{\wind_{\rth+(\nmalicious-k+1)}}} & \frac{1}{\omega_{E_2} - \omega_{\wind_{\rth+(\nmalicious-k+1)}}} & \hdots & \frac{1}{\omega_{E_k} - \omega_{\wind_{\rth+(\nmalicious-k+1)}}} &1\\ 
 \frac{1}{\omega_{E_1} - \omega_{\wind_{\rth+(\nmalicious-k+2)}}} & \frac{1}{\omega_{E_2} - \omega_{\wind_{\rth+(\nmalicious-k+2)}}} & \hdots & \frac{1}{\omega_{E_k} - \omega_{\wind_{\rth+(\nmalicious-k+2)}}} &1\\ 
 \vdots & \vdots & \ddots & \vdots & \vdots \\
 \frac{1}{\omega_{E_1} - \omega_{\wind_{\rth+(\nmalicious+1)}}} & \frac{1}{\omega_{E_2} - \omega_{\wind_{\rth+(\nmalicious+1)}}} & \hdots & \frac{1}{\omega_{E_k} - \omega_{\wind_{\rth+(\nmalicious+1)}}} &1
\end{pmatrix}
\end{equation}
\hrule
\end{figure*}

Now consider the column of $(\B_{\mset,\cdot}^\tp \mid \one_{(\nmalicious+1) \times 1})$ (in the rearranged form we just obtained) corresponding to worker $E_i$ ($i \leq k)$.
All entries in this column have the common non-zero factor $\left( \prod_{\ell \in \mathcal{R} \setminus \{E_i\}}(\omega_{E_i} - \omega_\ell) \right)^{-1}$. Therefore, the determinant of the lower-right $(k+1) \times (k+1)$ submatrix of the rearranged $(\B_{\mset,\cdot}^\tp \mid \one_{(\nmalicious+1) \times 1})$ is non-zero iff the determinant of the matrix in \cref{eq:reduced_matrix} is non-zero.
In \cref{eq:reduced_matrix}, we denote by $\omega_{\wind_{\rth+\groupind}}$ the evaluation point corresponding to the worker $\wind_{\rth+\groupind}$ of $\groupset$ that is not in the root group, cf. \cref{eq:grouping_scheme}.
The determinant of this matrix resembles the determinant of a Cauchy matrix (if the last column were removed, it would be a Cauchy matrix). The determinant of a Cauchy matrix is stated in \cite{schechterInversionCertainMatrices1959}; the above (Cauchy-like) determinant can be shown to be similar and non-zero by row and column reduction. Refer to \cref{sec: cauchy} for the proof.

\subsection{Pivot Analysis of a System of Linear Equations}
\label{sec: pivot}

For a full description of pivot analysis, refer to \cite{Treil_2017}.
Here, only the part on pivot analysis used in this work is replicated from \cite{Treil_2017} for the reader's convenience.
\begin{lemma}
    In a system of linear equations $\mathbf{Ax} = \mathbf{b}$, the system is inconsistent (has no solution) if and only if the echelon form of the augmented matrix $(\mathbf{A} \mid \mathbf{b})$ has a pivot in the last column.
\end{lemma}
\begin{proof}
    The backward direction is trivial since it would imply that there exists an equation consistent with the given set of equations such that $\mathbf{0}^\tp\mathbf{x} = 1$.
    For the forward direction, if we cannot obtain a row $(0, 0, \dots, 0 \mid 1)$ in the augmented matrix, we can just read off the solution from the row-reduced echelon form of the augmented matrix.
\end{proof}

\subsection{Determining the Combining Vector $\mathbf{b}$}
\label{sec: b-closedform}
Recall that each of the $\nworker$ workers is associated with a distinct element $\omega_\wind \in \galpha$.
\begin{lemma}
    Consider an index set $\mathcal{U}$ of size $\rth+1$ where $\rth = \nworker - (\nmalicious+\nhonest)$. With $\combvec$ given as in \cref{lemma:encoding-matrix}, the $\wind$-th entry of $\combvec$ (corresponding to worker $\wind$) is given by $\left( \prod_{\ell \in \groupset \setminus \{\wind\}}(\omega_\wind - \omega_\ell) \right)^{-1}$.
\end{lemma}
\begin{proof}
    Given a Vandermonde matrix
    \begin{equation*}
    \mathbf{V} =
    \begin{bmatrix}
        1 & x_1 & x^2_1 & \dots & x^{n-1}_1\\
        1 & x_2 & x^2_2 & \dots & x^{n-1}_2\\
        \vdots & \vdots & \vdots  & \ddots & \vdots \\
        1 & x_n & x^2_n & \dots &x^{n-1}_n
    \end{bmatrix}.
    \end{equation*}
    From~\cite{Rawashdeh2019ASM}, its inverse can be given by 
    \begin{equation*}
        v^{-1}_{i,j} = \frac{(-1)^{n-i}e_{n-i}(\{x_1,x_2,...,x_n\} \backslash\{x_j\})}{\prod^n_{m=1, m\neq j}(x_j - x_m)},
    \end{equation*}
    where
    \begin{equation*}
        e_m(\{y_1,...,y_k\}) = \sum_{1 \leq j_1 < ... < j_m\leq k}y_{j_1}...y_{j_m} \text{ for } m = 0,1,2,...,k.
    \end{equation*}
    Observe that the last row of the inverse of $\mathbf{V}$ can be given as $$v^{-1}_{n,j} = \frac{1}{\prod^n_{m=1, m\neq j}(x_j - x_m)}$$

    Recall the definition of $\combvec$ from \cref{eq:combvec}, which is the last column of $\MDS_{\cdot, {\groupset}}^{-1}$.
    Using this, and the fact that for any matrix $\mathbf{Z}$, $(\mathbf{Z}^{-1})^\tp = (\mathbf{Z}^\tp)^{-1}$ (inferring that the last column of the inverse of a matrix is the last row of the inverse of the transposed matrix), we obtain the result.
\end{proof}

\subsection{Determinant of Cauchy-like Matrix}
\label{sec: cauchy}

\begin{lemma}
    Any determinant of the form $$\begin{vmatrix}
 \frac{1}{\zeta_1 - \delta_1}& \frac{1}{\zeta_2 - \delta_1} & \hdots & \frac{1}{\zeta_k - \delta_{1}} &1\\ 
 \frac{1}{\zeta_1 - \delta_2}& \frac{1}{\zeta_2 - \delta_2} & \hdots & \frac{1}{\zeta_k - \delta_{2}} &1\\ 
 \vdots & \vdots & \ddots & \vdots & \vdots \\
 \frac{1}{\zeta_1 - \delta_{k+1}}& \frac{1}{\zeta_2 - \delta_{k+1}} & \hdots & \frac{1}{\zeta_k - \delta_{k+1}} &1 \notag
\end{vmatrix}$$

where $\zeta_1, \zeta_2, \dots, \zeta_{k}, \delta_1, \delta_2, \dots, \delta_{k+1}$ are distinct elements of any field and $1$ being the multiplicative identity, is non-zero.
\end{lemma}
\begin{proof}
    We use the following row and column reduction procedure (let $a_{ij}$ denote the entry in the $i$-th and $j$-th column):
    \begin{enumerate}
        \item Take out $-1$ from the first k columns. Now subtract the first column from all the other columns except the last. We now obtain:
        \begin{itemize}
            \item For $2 \leq j \leq k$:
            \begin{align*}
                a_{ij} &= \frac{1}{\delta_i - \zeta_j} - \frac{1}{\delta_i - \zeta_1}\\
                 & = \frac{\zeta_j - \zeta_1}{(\delta_i - \zeta_j)(\delta_i - \zeta_1)}
            \end{align*}
            \item and
            \begin{align*}
                a_{i,k+1} &= 1 - \frac{1}{\delta_i - \zeta_1}\\
                 &= \frac{\delta_i - \zeta_1 - 1}{\delta_i - \zeta_1}
            \end{align*}
            \item From columns $2$ to $k$ we can extract $\zeta_j - \zeta_1$ and we can extract $\frac{1}{\delta_i - \zeta_1}$ from each row.
        \end{itemize}
        The determinant is now of form:
        $$
        \begin{vmatrix}
         1& \frac{1}{\delta_1 - \zeta_2} & \dots & \frac{1}{\delta_{1} - \zeta_{k}}&\delta_1 - \zeta_1 - 1\\ 
         1& \frac{1}{\delta_2 - \zeta_2} & \dots & \frac{1}{\delta_{2} - \zeta_{k}}&\delta_2 - \zeta_1 - 1\\ 
         \vdots & \vdots & \ddots & \vdots & \vdots \\
        1 & \frac{1}{\delta_{k+1} - \zeta_2} & \dots  &\frac{1}{\delta_{k+1} - \zeta_{k}} & \delta_{k+1} - \zeta_1 - 1\notag
        \end{vmatrix}
        $$
        Observe that the multi-linearity of the determinant further allows us to remove the $-1$ from the last column.

        \item Subtract row 1 from all other rows. We obtain:
        \begin{itemize}
            \item $a_{1,j} = 0$
            \item $a_{i,k+1} = \delta_i - \delta_1$
            \item For $i\geq 2$ and $2\leq j \leq k$
            \begin{align*}
                a_{ij} &= \frac{1}{\delta_i - \zeta_j} - \frac{1}{\delta_1 - \zeta_j}\\
                &= \frac{\delta_1 - \delta_i}{(\delta_i - \zeta_j)(\delta_1 - \zeta_j)}
            \end{align*}
        \item Now extract $\delta_1 - \delta_i$ from all the rows except the first, $\frac{1}{\delta_1 - \zeta_j}$ from columns $2$ to $k$ and finally extract $-1$ from the last column.
        \end{itemize}
        The determinant is now of form:
        $$
        \begin{vmatrix}
         1& 1 &  1& \dots & 1&\delta_1 - \zeta_1 - 1\\ 
         0 & \frac{1}{\delta_2 - \zeta_2} & \frac{1}{\delta_2 - \zeta_3} & \dots &  \frac{1}{\delta_{2} - \zeta_{k}}&1\\ 
         0 & \frac{1}{\delta_3 - \zeta_2} & \frac{1}{\delta_3 - \zeta_3} & \dots &  \frac{1}{\delta_{3} - \zeta_{k}}&1\\ 
         \vdots & \vdots & \vdots & \ddots & \vdots & \vdots \\ 
        0 & \frac{1}{\delta_{k+1} - \zeta_2} &  \frac{1}{\delta_{k+1} - \zeta_3} & \dots &\frac{1}{\delta_{k+1} - \zeta_{k}} & 1\notag
        \end{vmatrix}
        $$
        \item Expanding the determinant along the first column, we obtain:
        $$
        \begin{vmatrix}
         \frac{1}{\delta_2 - \zeta_2} & \frac{1}{\delta_2 - \zeta_3} & \dots &  \frac{1}{\delta_{2} - \zeta_{k}}&1\\ 
         \frac{1}{\delta_3 - \zeta_2} & \frac{1}{\delta_3 - \zeta_3} & \dots &  \frac{1}{\delta_{3} - \zeta_{k}}&1\\ 
         \vdots & \vdots & \ddots & \vdots & \vdots \\ 
        \frac{1}{\delta_{k+1} - \zeta_2} &  \frac{1}{\delta_{k+1} - \zeta_3} & \dots &\frac{1}{\delta_{k+1} - \zeta_{k}} & 1\notag
        \end{vmatrix}
        $$
        \item We see the determinant is of identical form to what we started with. Hence, repeating the above steps, we finally obtain:
        $$
        \begin{vmatrix}
         \frac{1}{\delta_{k} - \zeta_{k}}&1\\ 
        \frac{1}{\delta_{k+1} - \zeta_{k}} & 1\notag
        \end{vmatrix}
        $$
        \item If the above determinant was zero, then $\delta_k = \delta_{k+1}$, which is a contradiction as they are two distinct elements. Thus, the determinant is non-zero.
    \end{enumerate}
\end{proof}

\subsection{Proof of \cref{rem: partialgradient}}
\label{proof: matchtree}

During a match, the main node obtains values for labels for all nodes along a path from the root to a leaf from two competing groups $\groupset[{\groupind_1}],\groupset[{\groupind_2}]$. This is possible due to \cref{rem: encoding-matrix} since both groups contain $\rth + 1$ workers. Each worker encodes the computed partial gradients according to the requested partial sum into $\iresp_{\wind}^{(\mathbf{a})}$. The main node obtains the claimed partial sums as $\Z_{\groupset[{\groupind_1}]}^{(\mathbf{a})} \combvec[{\groupind_1}]$ and $\Z_{\groupset[{\groupind_2}]}^{(\mathbf{a})} \combvec[{\groupind_2}]$. If these values differ, there must be a child node for which the labels differ. If the groups agree on the left child node, they must differ on the right child node and vice versa. Observe that by knowing the label of one of the child nodes, the main node can infer the label of the other child node since the parent node's label is known. The tournament then advances to one of the child nodes in the match tree, again decoding the claimed labels. Since the two groups disagree on the root node, we see by induction that there is a path from the root to the leaf such that the groups disagree on all the labels in this path and, thus, on the leaf. By \cref{rem: encoding-matrix}, we see that each worker has committed to a claimed value of the single partial gradient corresponding to this leaf since when the main node requests a single partial gradient from a group, each worker is requested only to send their own claimed value for that gradient. This claimed value is zero if the worker is not assigned the partial gradient to compute. The main server finally obtains a contradicting partial gradient $\pgrad$ corresponding to this leaf node from all workers in $\groupset[{\groupind_1}]$ and $\groupset[{\groupind_2}]$ that have computed this partial gradient and locally computes $\pgrad$.